\documentclass[a4paper]{llncs}
\usepackage{graphicx,latexsym}

\begin{document}

\title{On the Parameterized Complexity of Associative and Commutative Unification\thanks{
This work was partially supported by the Collaborative Research Programs
of National Institute of Informatics.}}
\titlerunning{}

\author{Tatsuya Akutsu\inst{1}
\and Takeyuki Tamura\inst{1}
\and 
Atsuhiro Takasu\inst{2}}

\institute{Bioinformatics Center, Institute for Chemical Research,
Kyoto University,\\
Gokasho, Uji, Kyoto, 611-0011, Japan.\\
\email{\{takutsu, tamura\}@kuicr.kyoto-u.ac.jp}\\
\and
National Institute of Informatics, Tokyo 101-8430, Japan.\\
\email{takasu@nii.ac.jp}}

\maketitle

\begin{abstract}
This paper studies the unification problem with
associative, commutative, and associative-commutative functions mainly
from a viewpoint of the parameterized complexity on the number of
variables.
It is shown that
both associative and associative-commutative unification problems
are $W[1]$-hard.
A fixed-parameter algorithm and a polynomial-time algorithm are presented for
special cases of commutative unification in which
one input term is variable-free and the number of variables is bounded by 
a constant, respectively.
Related results including those on
the string and tree edit distance problems with variables are shown too.
\end{abstract}

\section{Introduction}

\emph{Unification} plays an important role in various areas of computer science,
including theorem proving, logic programming, natural language processing,
and database query systems \cite{kapur92,knight89}.
The unification problem is, in the fundamental form,
to find a substitution for variables that will make the two given
terms identical, where terms are built up from function symbols, 
variables, and constants \cite{knight89}.
For example, two terms $f(x,y)$ and $f(g(a),f(b,x))$ become identical
by substituting $x$ and $y$ by $g(a)$ and $f(b,g(a))$, respectively.
If one input term contains no variable, the problem is called \emph{matching}.

Although unification has a long history beginning from a seminal work by
Herbrand in 1930 (e.g., see \cite{knight89}),
it is becoming important again because 
math search recently attracts researchers in information retrieval
(IR) community \cite{kamali10,kim12,nguyen12}.
For example, math search is adopted as a pilot task in an IR evaluation
conference
NTCIR \footnote{http://research.nii.ac.jp/ntcir/ntcir-10/conference.html}.
The math search is a sort of IR task to retrieve documents
containing mathematical formulas
and/or formulas themselves similar to a query.
Several systems have been developed such as
Wolfram Formula Search\footnote{http://functions.wolfram.com/formulasearch} and
formula search for Wikipedia\footnote{ http://shinh.org/wfs/}.
Since mathematical formulas are usually represented
with tree structures, structural similarity is important
to measure the similarity between formulas.
Approximate tree matching \cite{bille05} is
a key to measure the similarity.
However, when measuring the similarity between mathematical formulas,
we need to unify substitution of variables.
For example, a query $x^2 + x$ has same similarity to formulas
$y^2 + z$ and $y^2 + y$ by tree edit distance,
although these two formulas are mathematically different.
Therefore, approximate tree matching is not enough and 
combination with unification is strongly needed.

Returning to unification, many variants have been proposed
\cite{benanav87,kapur92,knight89}.
Among them, unification with commutative and associative functions
are important from the viewpoint of math search
because many functions satisfy either one or both of these two properties,
where functions satisfying $f(x,y)=f(y,x)$ and $f(x,f(y,z))=f(f(x,y),z)$
are called \emph{commutative} and \emph{associative}, respectively.

Extensive studies have been done on the computational complexity
of various unification problems.
For the fundamental one, beginning from Robinson's exponential time
algorithm \cite{robinson65},
a linear time algorithm was finally developed \cite{paterson78}.
However, all of associative, commutative, and associative-commutative
unification (and matching) problems are known to be
NP-hard \cite{benanav87,eker02,kapur92}.
Polynomial time algorithms are known only for very restricted cases
\cite{aikou05,benanav87,kapur92}.
For example, it is known that associative-commutative matching
can be done in polynomial time if
every variable occurs only once \cite{benanav87}.
From a practical viewpoint, many studies have been done on
various extensions of unification.
Furthermore, combination with approximate tree matching 
has been studied \cite{gilbert00,iranzo10}.
However, these are heuristic algorithms.

In this paper, we study associative, commutative, and associative-commutative
unification mainly from a viewpoint of parameterized complexity on
the number of variables because
the number of variables is often much smaller than the size of terms.
We show the following results along with related results:
(i) both associative and
associative-commutative matching problems are $W[1]$-hard,
(ii) both associative and associative-commutative unification 
can be done in polynomial time if every variable occurs only once,
(iii) commutative matching can be done in $O(2^{k} poly(m,n))$ time
where $k$ is the number of variables and $m,n$ are the size of
input terms,
(iv) commutative unification can be done in polynomial time
if the number of variables is bounded by a constant.
In addition, we show that 
both the string and tree edit distance problems with variables 
are $W[1]$-hard.
All algorithms presented in this paper simply decide whether two terms
are unifiable and do not output the corresponding substitutions.
However, the algorithms can be modified to output such substitutions
(when unifiable)
by using the standard \emph{traceback} technique.

\section{String Edit Distance with Variables}

Let $\Sigma$ be an alphabet and $\Gamma$ be a set of variables,
where we mainly consider $\Gamma$ that is defined as
the set of variables appearing in the input.
A string is a sequence of symbols over $\Sigma \cup \Gamma$.
Let $\theta$ be a \emph{substitution},
which is a mapping from $\Gamma$ to $\Sigma$.
For a string $s$, let $s \theta$ denote the string over $\Sigma$
obtained by replacing all occurrences of variables $x \in \Gamma$
by $\theta(x)$.
We call two strings $s_1$ and $s_2$ are \emph{unifiable}
if there exists a substitution $\theta$ such that
$s_1 \theta = s_2 \theta$.

\begin{example}
Let $s_1 = abcxbcx$, $s_2 = abydbzd$, and $s_3 = abydbzc$,
where $\Gamma=\{x,y,z\}$.
Then, $s_1$ and $s_2$ are unifiable since
$s_1 \theta = s_2 \theta = abcdbcd$ holds
by $\theta=\{x/d,y/c,z/c\}$.\footnote{$x/a$ means that $x$ is substituted by
$a$.}
However, $s_1$ and $s_3$ are not unifiable since there does not exist
$\theta$ such that $s_1 \theta = s_3 \theta$.
\end{example}

For string $s$ and integer $i$, $s[i]$ denotes the $i$-th character of $s$,
$s[i \ldots j]$ denotes $s[i] \ldots s[j]$, and
$|s|$ denotes the length (i.e., the number of characters) of $s$.
For two strings $s$ and $t$ (including the case of $s$ and/or $t$ are
single characters),
$st$ denotes the string obtained by concatenating $s$ and $t$.
An {\it edit operation on a string} $s$ over $\Sigma$ is either
a {\it deletion}, an {\it insertion}, 
or a {\it replacement} of a character of $s$ \cite{bodlaender95}.\footnote{
Usually, `substitution' is used instead of `replacement'.
However, we use `replacement' in order to discriminate
from `substitution' in unification.}
The {\it edit distance between two strings} $s_1$ and $s_2$ over $\Sigma$
is defined as
the minimum number of operations to transform $s_1$ to $s_2$,
where we consider unit cost operations here.
Let $d_S(s_1,s_2)$ denote the edit distance 
between $s_1$ and $s_2$.
From the definition, 
$
d_S(s_1,s_2) ~=~ \min_{ed:ed(s_1)=s_2} |ed| ~=~ \min_{ed:ed(s_2)=s_1} |ed|
$  
holds where $ed$ is a sequence of edit operations.
For example, $d_S(bcdfe,abgde)=3$ because $abgde$ is obtained from
$bcdfe$ by deletion of $f$, replacement of $c$ to $g$, and insertion of $a$. 
We also define the edit distance $\hat{d}_S$
between two strings over $\Sigma \cup \Gamma$
by
$
\hat{d}_S(s_1,s_2) =
\min_{ed:(\exists \theta)(ed(s_1)\theta = s_2 \theta)} |ed|.
$
This variant of edit distance is called \emph{edit distance with variables}.
Although it is well known that $d_S(s_1,s_2)$ can be computed
in polynomial time, computation of $\hat{d}_S(s_1,s_2)$ is $W[1]$-hard
as shown below.

\begin{theorem}
The edit distance problem with variables is $W[1]$-hard with respect to
the number of variables even if the number of occurrences of each variable
is bounded by 3.
\label{thm:edhard}
\end{theorem}
\begin{proof}
We present an \emph{FPT-reduction} \cite{flum06}
from the longest common subsequence problem (LCS).
LCS is, given a set of strings $\{s_1,\ldots,s_k\}$ over $\Sigma_0$ and 
an integer $l$,
to decide whether there exists a string $s$ of length $l$
that is a subsequence of each string $s_i$.
where $s$ is called a \emph{subsequence} of $s'$ if $s$ is obtained by
deletion operations from $s'$.
It is known that LCS is $W[1]$-hard
for parameters $k$ and $l$ \cite{bodlaender95}.

First we consider the case in which there is no constraint on the number
of occurrences of variables.
From an instance of LCS, we construct an instance of
edit distance with variables as follows.
Let $\Sigma=\Sigma_0 \cup \{ \# \}$ and $\Gamma=\{x_1,x_2,\ldots,x_l\}$,
where $\#$ is a symbol not appearing in $s_1,\ldots,s_k$.
We construct $s^1$ and $s^2$ by
\begin{eqnarray*}
s^1 & = & x_1 x_2 \ldots x_l \# x_1 x_2 \ldots x_l \# \cdots \#
x_1 x_2 \ldots x_l,\\
s^2 & = & s_1 \# s_2 \# \cdots \# s_k,
\end{eqnarray*}
where $x_1 x_2 \ldots x_l$ appears $k$ times in $s^1$.
Then, we can see that
there exists an LCS of length $|s^1|$ iff
$\hat{d}_S(s^1,s^2) = (\sum_i |s_i|)-lk$
(i.e., there exists $\theta$ such that $s^1 \theta$ is a subsequence of
$s^2$).
Since the number of variables appearing in this instance is $l$,
it is an FPT reduction.
The proof for the case of the bounded number of occurrences 
is given in Appendix A1.
\qed
\end{proof}

If each variable occurs only once, then the problem is equivalent
to approximate string matching with don't care characters,
which can be solved in polynomial time \cite{akutsu95}.
It should be noted that if an alphabet $\Sigma$ is fixed,
the number of possible $\theta$ is bounded by $|\Sigma|^k$,
where $k=|\Gamma|$.
Therefore, we have a fixed-parameter algorithm with parameter $k$
for a fixed alphabet.

\begin{proposition}
The edit distance problem with variables can be solved in\\
$O(|\Sigma|^k poly(m,n))$ time where $k$ is the number of variables
and $m$ and $n$ are the size of input strings.
\end{proposition}

\section{Unification}

In order to define unification,
we regard $\Sigma$ as a set of function symbols, where 
arity (i.e., the number of arguments) is associated with each symbol.
We call a function symbol with arity 0 a \emph{constant}.
We define a \emph{term} as follows:
\begin{itemize}
\item a constant is a term,
\item a variable is a term,
\item if $t_1,\cdots,t_d$ are terms and
$f$ is a function symbol with arity $d>0$,
$f(t_1,\ldots,t_d)$ is a term.
\end{itemize}
We identify each term $t$ with a rooted ordered tree where
each node corresponds to a function symbol and each leaf
corresponds to a constant.
For a term $t$, $N(t)$ denotes the set of nodes in a tree $t$,
$r(t)$ denotes the root of $t$,
and $\gamma(t)$ denotes the function symbol of $r(t)$.
For a node $u \in N(t)$, $t_u$ denotes a subterm (i.e., subtree)
of $t$ rooted at $u$.
The \emph{size} of $t$ is defined as $|N(t)|$.

Let ${\cal T}$ be a set of terms over $\Sigma$ and $\Gamma$.
Then, a \emph{substitution} $\theta$ is defined as 
a (partial) mapping from $\Gamma$ to ${\cal T}$,
where $t$ must not contain a variable $x$ if $x/t \in \theta$.
For a term $t$ and a substitution $\theta$,
$t \theta$ is the term obtained by simultaneously replacing
variables according to $\theta$.
We say that terms $t_1$ and $t_2$ are \emph{unifiable}
if there exists $\theta$ such that $t_1 \theta = t_2 \theta$.
Such $\theta$ is called a \emph{unifier}.
In this paper, the unification problem is to decide
whether two given terms are unifiable and output a unifier
if unifiable.\footnote{In the standard case, it is required to
output the most general unifier (mgu). However, in most variants
considered in this paper, there does not exist mgu.}
It is well-known that the unification problem can be solved in
linear time \cite{paterson78}.
A special case of the unification problem in which $t_2$ is
variable-free is called a \emph{matching} problem.
If every variable (resp., a variable $x$) occurs in a term $t$ only once,
the term (resp., the variable) is called a \emph{DO-term} (resp.,
\emph{DO-variable}), where DO means distinct occurrence(s) \cite{benanav87}.
Unless otherwise stated, $m$ and $n$ denote the size of two input terms 
$t_1$ and $t_2$.

\begin{example}
Let $t_1 = f(g(a,b,a),f(x,x))$,
$t_2 = f(g(y,b,y),z)$,
$t_3 = f(g(a,b,u),f(v,u))$,
$t_4 = f(f(a,b),f(a,a))$,
$t_5 = f(g(a,b,a),f(w,f(w,w)))$,
where $\Gamma=\{x,y,z,u,v,w\}$.
$t_1$ and $t_2$ are unifiable since 
$t_1 \theta_1 = t_2 \theta_1 = f(g(a,b,a),f(x,x))$ holds for
$\theta_1 = \{y/a,z/f(x,x)\}$.
$t_1$ and $t_3$, and $t_2$ and $t_3$ are also unifiable since 
$t_1 \theta_2 = t_3 \theta_2 = f(g(a,b,a),f(a,a))$ and
$t_2 \theta_3 = t_3 \theta_3 = f(g(a,b,a),f(v,a))$ hold for
$\theta_2=\{x/a,u/a,v/a\}$ and
$\theta_3=\{y/a,u/a,z/f(v,a)\}$, respectively.
$t_4$ is not unifiable to $t_1$, $t_2$, or $t_3$.
$t_5$ is unifiable to $t_2$, but is not unifiable to
$t_1$ (or $t_3$) because it is impossible to simultaneously satisfy
$x=w$ and $x=f(w,w)$.\footnote{This kind of matching can be avoided by
\emph{occur check}.}
\end{example}

As in the case of string edit distance,
we can combine \emph{tree edit distance} \cite{bille05}
with unification.
Let $d_T(t_1,t_2)$ denote the tree edit distance where
the distance can be for both ordered and unordered trees.
Then, we define the tree edit distance $\hat{d}_T$
between two trees (i.e., two terms) over $\Sigma \cup \Gamma$
by
\[
\hat{d}_T(t_1,t_2) = \min_{ed:(\exists \theta)(ed(t_1)\theta = t_2 \theta)} |ed|.
\]
By combining the proofs of Thm.~\ref{thm:edhard} and
Thm.~\ref{thm:asohard}, we have:

\begin{theorem}
The tree edit distance problem
with variables is $W[1]$-hard for both ordered and unordered trees
with respect to
the number of variables for a fixed alphabet even if
the number of occurrences of each variable is bounded by 3.
\label{thm:tree-ed-hard}
\end{theorem}

We also have the following theorem as in several matching problems
\cite{benanav87},
where the proof is given in Appendix A2.\footnote{
The original (i.e., variable-free) edit distance problem is known to NP-hard
for unordered trees \cite{bille05}.}

\begin{theorem}
The ordered tree edit distance problem with variables
can be solved in polynomial time for DO-terms.
\label{thm:tree-ed-do}
\end{theorem}

\section{Associative Unification}

A function $f$ is called \emph{associative} if
$f(x,f(y,z))=f(f(x,y),z)$ always holds.
Associative unification is a variant of unification
in which some functions are associative.
In this section, we assume that all functions are associative
although all the results are valid even if usual functions are included.

It is shown that associative matching is NP-hard \cite{benanav87}.
However, the proof in \cite{benanav87} does not work to show
the parameterized hardness.

\begin{theorem}
Associative matching is $W[1]$-hard with respect to the number of variables
even for a fixed $\Sigma$.
\label{thm:asohard}
\end{theorem}
\begin{proof}
As in the proof of Thm.~\ref{thm:edhard},
we use a reduction from LCS (see also Fig.~\ref{fig:asohard}).

First we consider an infinite alphabet.
Let $(\{s_1,\ldots,s_k\},l)$ be an instance of LCS.
For each $i=1,\ldots,k$, we create a term $u^i$ by
\[
u^i = 
f(y_{i,1},f(x_1,f(y_{i,2},f(x_2,\cdots f(y_{i,l},f(x_l,f(y_{i,l+1},g(\#,\#)))
\cdots )))),
\]
where $\#$ is a character not appearing in $s_1,\ldots,s_k$.
We create a term $t_1$ by concatenating $u^1,\ldots,u^k$, which can be done
by replacing the last occurrence of $\#$ of each
$u^i$ by $u^{i+1}$ for $i=1,\ldots,k-1$.
Then, we transform each $s_i$ into a string $s_i'$ of length
$1 + 2 \cdot |s_i|$ by inserting a special character $\&$ 
in front of each character in $s_i$, and appending $\&$
at the end of $s_i$,
where each $\&$ is considered as a distinct constant
(i.e., $\&$ cannot match any symbol, but can match any variable).
We represent each $s_i'$ by a term $t^i$ given as
\[
t^i =
f(s_i'[1],f(s_i'[2],f(s_i'[3],f(\cdots,f(s_i'[1+2\cdot |s_i'|],g(\#,\#))
\cdots)))).
\]
We create a term $t_2$ by concatenating $t^1,\ldots,t^k$.

\begin{figure}[ht]
\begin{center}
\includegraphics[width=7cm]{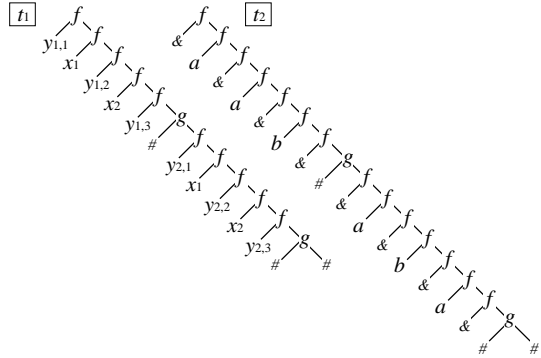}
\caption{Example of a reduction for the case of $s_1=aab$, $s_2=aba$,
and $l=2$ in the proof of Thm.~\ref{thm:asohard}.}
\label{fig:asohard}
\end{center}
\end{figure}

Then, we can see that $t_1$ and $t_2$ are unifiable
iff there exists an LCS of length $l$.
Since the number of variables in $t_1$ is $(l+1)k+l=lk+l+k$,
it is an FPT-reduction and thus the problem is $W[1]$-hard.

Finally, we represent each constant by a distinct term using
a function symbol $h$ and binary-encoding 
(e.g., 10-th symbol (among 16 symbols) can be represented as
$h(1,h(0,h(1,0)))$ ).
\qed
\end{proof}

Next, we consider associative unification for DO-terms,
where it has some similarity with DO-associative-commutative matching
\cite{benanav87}.
We begin with the simplest case in which each term does not
contain a variable.

\begin{proposition}
Associative unification can be done in polynomial time if
there exists no variable.
\end{proposition}
\begin{proof}
We transform each input term into its \emph{canonical form} in which 
consecutive and same function symbols are simplified into one symbol.
For example, both $f(f(a,b),f(g(c,d),e))$ and $f(a,f(b,f(g(c,d),e)))$
are transformed into $f(a,b,g(c,d),e)$.
Since $t_1 = t_c$ and $t_2 = t_c$ means $t_1 = t_2$,
it is enough to test the isomorphism of the canonical forms in order to
examine $t_1 = t_2$.
Since the rooted ordered tree isomorphism between the resulting
canonical forms can be trivially tested in linear time,
we have the proposition.
\qed
\end{proof}

In order to treat DO-terms,
we transform terms $t_1$ and $t_2$ into their canonical forms $t^1$ and $t^2$.
Then, we apply the following procedure to $t^1$ and $t^2$
(see also Fig.~\ref{fig:asopoly}),
where it returns `true' iff $t^1$ and $t^2$ are unifiable, and
$D[u,v]=1$ iff $(t^1)_u$ and $(t^2)_v$ are unifiable.
It is to be noted that step (\#) can be done in constant time because
$(t^1)_u$ and $(t^2)_v$ are unifiable there iff
these are the same constant or one of $t_u$ and $t_v$ is a variable.

\begin{rm}
\begin{tabbing}
\quad \= \quad \= \quad \= \quad \= \quad \= \quad \= \quad \= \kill
\> Procedure $AssocMatchDO(t^1,t^2)$\\
\> \> {\bf for all} $u \in N(t^1)$ {\bf do}~~~~~~~~~~~~~~~~~~~~~~~~~/* in post-order */ \\
\> \> \> {\bf for all} $v \in N(t^2)$ {\bf do}~~~~~~~~~~~~~~~~~~~~~~/* in post-order */ \\
\> \> \> \> {\bf if} $(t^1)_u$ or $(t^2)_v$ is a constant {\bf then}\\
\> \> \> \> \> {\bf if} $(t^1)_u$ and $(t^2)_v$ are unifiable ~~~~~~~~~~~~~~~~~~~~~~-(\#) \\
\> \> \> \> \> {\bf then} $D[u,v] \leftarrow 1$ {\bf else} 
$D[u,v] \leftarrow 0$;\\
\> \> \> \> {\bf else if} $(t^1)_u$ or $(t^2)_v$ is a variable {\bf then}\\
\> \> \> \> \> $D[u,v] \leftarrow 1$;\\
\> \> \> \> {\bf else}~~~/* $(t^1)_u = f_1((t^1)_{u_1},\ldots,(t^1)_{u_p})$,
$(t^2)_v = f_2((t^2)_{v_1},\ldots,(t^2)_{v_q})$ */\\
\> \> \> \> \> {\bf if} $f_1 = f_2$ and 
$\langle (t^1)_{u_1},\ldots,(t^1)_{u_p} \rangle$ can match $\langle (t^2)_{v_1},\ldots,(t^2)_{v_q} \rangle$  \\
\> \> \> \> \> {\bf then} $D[u,v] \leftarrow 1$ {\bf else} $D[u,v] \leftarrow 0$; \\
\> \> {\bf if} $D[r(t_1),r(t_2)]=1$
{\bf then} {\bf return} true {\bf else} {\bf return} false.
\end{tabbing}
\end{rm}

Match of $\langle (t^1)_{u_1},\ldots,(t^1)_{u_p} \rangle$ 
and $\langle (t^2)_{v_1},\ldots,(t^2)_{v_q} \rangle$ can be tested in
polynomial time by regarding each of these two sequences as a string
and applying string matching with variable length don't cares
\cite{akutsu96} with setting the difference to be 0 and allowing don't
care characters appear in both strings,
where
$(t^1)_{u_i}$ (resp., $(t^2)_{v_j}$) is regarded as a don't care symbol
that can match any substring of length at least 1 if it is a variable,
otherwise $(t^1)_{u_i}$ can match $(t^2)_{v_j}$ iff $D[u_i,v_j]=1$
(see Appendix A3 for the details).
Since for-loops are repeated $O(mn)$ times and
string matching with variable length don't cares can be done in polynomial time,
we have:

\begin{theorem}
Associative unification for DO-terms can be done in polynomial time.
\label{thm:asspoly}
\end{theorem}

\begin{figure}[ht]
\begin{center}
\includegraphics[width=11cm]{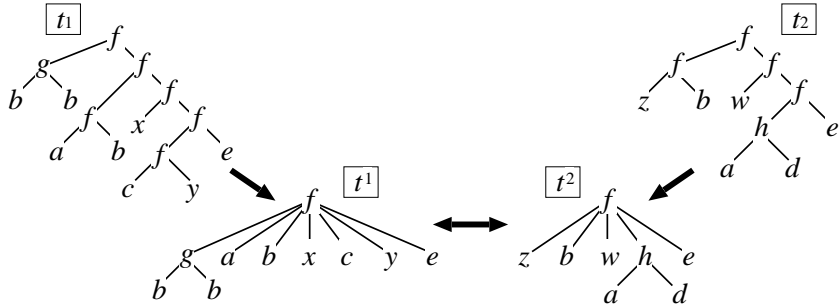}
\caption{Illustration of associative unification for DO-terms $t_1$ and $t_2$.
$t_1$ and $t_2$ are transformed into $t^1$ and $t^2$,
which are then unified by
$\theta=\{y/h(a,d),z/f(g(b,b),a),w/f(x,c)\}$.}
\label{fig:asopoly}
\end{center}
\end{figure}

We can also consider another variant in which
$t_1$ can contain a constant number of non-DO variables
but $t_2$ cannot contain any variable.
Let $\Gamma_1$ be the set of non-DO variables in $t_1$.
We examine all possible mappings from $\Gamma_1$ to the set of
consecutive children of each node in the canonical form $t^2$ of $t_2$.
If we apply such a mapping, all occurrences of variables in $\Gamma_1$
are replaced by terms without variables.
Then, we can apply $AssocMatchDO(t^1,t^2)$ to the resulting terms.
Since the number of mappings is clearly $O(n^{2 |\Gamma_1|})$,
we have:

\begin{corollary}
Associative matching can be done in polynomial time 
if $t_1$ contains a constant number of non-DO terms and
any number of DO-terms.
\label{cor:asspoly}
\end{corollary}

\section{Commutative Unification}

A function $f$ is called \emph{commutative} if
$f(x,y)=f(y,x)$ always holds.
Commutative unification is a variant of unification
in which some functions are commutative.
It is known that 
even commutative matching
is NP-hard (by a reduction from 3SAT) \cite{benanav87}.
In this section,
we present a parameterized algorithm for commutative matching
and a polynomial-time algorithm for commutative unification
with a bounded number of variables.

First we note that commutative unification can be done in polynomial time
if both $t_1$ and $t_2$ are variable-free because it is equivalent to
the rooted unordered tree isomorphism problem.

\begin{proposition}
Commutative unification can be done in polynomial time
if both $t_1$ and $t_2$ are variable-free.
\end{proposition}

Next we consider commutative matching.
We construct a 0-1 table $D[u,v]$ for node pairs
$(u,v) \in N(t_1) \times N(t_2)$
by applying dynamic programming in a bottom-up manner,
where $D[u,v]=1$ iff $(t_1)_u$ is unifiable to $(t_2)_v$.
It is enough to construct such table entries only for pairs
with the same depth.
We also construct a table $\Theta[u,v]$, where each element
holds a set of possible substitutions $\theta$
such that $(t_1)_u \theta = (t_2)_v$.

Let $\theta_1 = \{x_{i_1}/t_{i_1},\ldots,x_{i_p}/t_{i_p} \}$ and
$\theta_2 = \{x_{j_1}/t_{j_1},\ldots,x_{j_p}/t_{j_q} \}$ be
substitutions.
$\theta_1$ is said to be \emph{compatible} with
$\theta_2$ if there exists no variable $x$ such that
$x = x_{i_a} = x_{j_b}$ but 
$t_{i_a} \neq t_{j_b}$.
Let $\Theta_1$ and $\Theta_2$ be sets of substitutions.
We define $\Theta_1 \Join \Theta_2$ by
\[
\Theta_1 \Join \Theta_2 ~=~  \{
\theta_i \cup \theta_j |~
\theta_i \in \Theta_1 \mbox{ is compatible with }
\theta_j \in \Theta_2 \}.
\]
For a node $u$, $u_L$ and $u_R$ denote the left and right children of
$u$, respectively.

\begin{rm}
\begin{tabbing}
\quad \= \quad \= \quad \= \quad \= \quad \= \quad \= \quad \= \kill
\> Procedure $CommutMatch(t_1,t_2)$\\
\> \> {\bf for all} pairs $(u,v) \in N(t_1) \times N(t_2)$ with the same depth\\
\> \> {\bf do} \hspace{6cm} /* in a bottom-up way */ \\
\> \> \> {\bf if} $(t_1)_u$ is a variable {\bf then} \\
\> \> \> \> $\Theta[u,v] \leftarrow \{ \{(t_1)_u/(t_2)_v\} \}$;
$D[u,v] \leftarrow 1$ \\
\> \> \> {\bf else if} $(t_1)_u$ does not contain a variable {\bf then}\\ 
\> \> \> \> $\Theta[u,v] \leftarrow \emptyset$; \\
\> \> \> \> {\bf if} $(t_1)_u = (t_2)_v$ {\bf then} $D[u,v] \leftarrow 1$
{\bf else} $D[u,v] \leftarrow 0$ \\
\> \> \> {\bf else if} $\gamma((t_1)_u) \neq \gamma((t_2)_v)$ {\bf then}\\
\> \> \> \> $\Theta[u,v] \leftarrow \emptyset$; $D[u,v] \leftarrow 0$ \\
\> \> \> {\bf else} \\
\> \> \> \> $\Theta[u,v] \leftarrow \emptyset$; $D[u,v] \leftarrow 0$; \\
\> \> \> \> {\bf for all} $(u_1,u_2,v_1,v_2) \in
\{ (u_L,u_R,v_L,v_R), (u_R,u_L,v_L,v_R) \}$ {\bf do} ~~~~~ -(\#)\\
\> \> \> \> \> {\bf if} $D[u_1,v_1]=1$ and $D[u_2,v_2]=1$ and
$\Theta_1[u_1,v_1] \Join \Theta_2[u_2,v_2] \neq \emptyset$\\
\> \> \> \> \> {\bf then} $\Theta[u,v] \leftarrow \Theta[u,v] \cup 
(\Theta_1[u_1,v_1] \Join \Theta_2[u_2,v_2])$;
$D[u,v] \leftarrow 1$;\\
\> \> {\bf if} $D[r(t_1),r(t_2)]=1$ {\bf then return} true {\bf else return} false.
\end{tabbing}
\end{rm}

\begin{theorem}
Commutative matching can be done in $O(2^{k} poly(m,n))$ time where
$k$ is the number of variables in $t_1$.
\end{theorem}
\begin{proof}
The correctness follows from the observation that
each variable is substituted by a term without variables and
$f(x,y)=f(y,x)$ is taken into account at step (\#).

In order to analyze the time complexity,
we consider the size (i.e., the number of elements) 
of $\Theta[u,v]$.
An important observation is that
if $(t_1)_{u_L}$ does not
contain a variable,
$|\Theta[u,v]| \leq \max(|\Theta[u_R,v_L]|,|\Theta[u_R,v_R]|)$ holds
(an analogous property holds for $(t_1)_{u_R}$).
Let $B_i$ denote the maximum size of $\Theta[u,v]$ when
the number of (distinct) variables in $(t_1)_u$ is $i$.
Then, we can see that the following relations hold:
$
B_1 = 1,~~B_{i+j} =  2 B_i B_j,
$
from which $B_i = 2^{i-1}$ follows.
Therefore, computation of
$\Theta_1[u_1,v_1] \Join \Theta_2[u_2,v_2]$ can be done in
$O(2^{k} poly(m,n))$ time by using `sorting' as in usual `join' operations.
Then, we can see that the total computation time is also
$O(2^{k} poly(m,n))$.
\qed
\end{proof}

Next, we consider the case where both $t_1$ and $t_2$ contain variables.
As in the case of linear time unification \cite{paterson78},
we assume that two variable free terms $t_1$ and $t_2$
are represented by a DAG (directed acyclic graph) $G(V,E)$,
where $t_1$ and $t_2$ respectively correspond to $r_1$ and $r_2$ 
of indegree 0 ($r_1,r_2 \in V$).
Then, whether $r_1$ and $r_2$ represent the same term can be tested 
in polynomial time with respect to the size of $G$
by using the following procedure, where $t_u$ denotes the term
corresponding to a node $u$ in $G$.

\begin{rm}
\begin{tabbing}
\quad \= \quad \= \quad \= \quad \= \quad \= \quad \= \quad \= \quad \= \quad \= \kill
\> Procedure $TestCommutIdent(r_1,r_2,G(V,E))$\\
\> \> {\bf for all} $u \in V$ {\bf do}~~~~~~~~~~~~~~~~~~~~~~~~~~~~~~~~~~~~~~~~~~~~~/* in post-order */\\
\> \> \> {\bf for all} $v \in V$ {\bf do}~~~~~~~~~~~~~~~~~~~~~~~~~~~~~~~~~~~~~~~~~~/* in post-order */\\
\> \> \> \> {\bf if} $u=v$ {\bf then} $D[u,v] \leftarrow 1$; {\bf continue};\\
\> \> \> \> {\bf if} $t_u$ or $t_v$ is a constant {\bf then} \\
\> \> \> \> \> {\bf if} $t_u = t_v$ {\bf then} $D[u,v] \leftarrow 1$ {\bf else} $D[u,v] \leftarrow 0$;\\
\> \> \> \> {\bf else}\\
\> \> \> \> \> Let $u = f_1(u_L,u_R)$ and $v=f_2(v_L,v_R)$; \\
\> \> \> \> \> {\bf if} $f_1 = f_2$ {\bf then}\\
\> \> \> \> \> \> {\bf if} ($D[u_L,v_L]=1$ and $D[u_R,v_R]=1$) or \\
\> \> \> \> \> \> \> \> \> ($D[u_L,v_R]=1$ and $D[u_R,v_L]=1$) \\
\> \> \> \> \> \> {\bf then} $D[u,v] \leftarrow 1$ {\bf else} $D[u,v] \leftarrow 0$\\
\> \> \> \> \> {\bf else} $D[u,v] \leftarrow 0$;\\
\> \> {\bf if} $D[r_1,r_2]=1$ {\bf then} {\bf return} true {\bf else} {\bf return} false.
\end{tabbing}
\end{rm}

In order to cope with terms with variables,
we consider all possible mappings from the set of variables to
$N(t_1) \cup N(t_2)$.
For each mapping,
we replace all appearances of the variables by the corresponding nodes,
resulting in a DAG to which we can apply $TestCommutIdent(r_1,r_2,G(V,E))$.
The following is a pseudo-code of the procedure for terms with variables.

\begin{rm}
\begin{tabbing}
\quad \= \quad \= \quad \= \quad \= \quad \= \quad \= \quad \= \kill
\> Procedure $CommutUnify(t_1,t_2)$\\
\> \> {\bf for all} mappings $M$ from a set of variables to nodes in $t_1$ and $t_2$ {\bf do} \\
\> \> \> {\bf if} there exists a directed cycle (excluding a self-loop) {\bf then} {\bf continue};\\
\> \> \> Replace each variable having a self-loop with a distinct constant symbol;\\
\> \> \> Replace each occurrence of a variable node $u$ with node $M(u)$;\\
\> \> \> \> \> \> /* if  $M(u)=v$ and $M(v)=w$, $u$ is replaced by $w$ */\\ 
\> \> \> Let $G(V,E)$ be the resulting DAG;\\
\> \> \> Let $r_1$ and $r_2$ be the nodes of $G$ corresponding to $t_1$ and $t_2$;\\
\> \> \> {\bf if} $CommutIdent(r_1,r_2,G(V,E))=$true {\bf then} {\bf return} true;\\
\> \> {\bf return} false.
\end{tabbing}
\end{rm}

\noindent
Then, we have the following, where the proof is given in Appendix A4.

\begin{theorem}
Commutative unification can be done in polynomial time
if the number of variables in $t_1$ and $t_2$ is 
bounded by a constant.
\label{thm:com-unif-poly}
\end{theorem}

\section{Associative-Commutative Unification}

Associative-commutative unification is a variant of unification
in which some functions can be both associative and commutative.
We show that associative-commutative matching is
$W[1]$-hard even if all every function is associative and commutative,
where the proof is a bit involved and is given in Appendix A5.

\begin{theorem}
Matching is $W[1]$-hard with respect to the number of variables
even if every function symbol is associative and commutative.
\label{thm:ac-hard}
\end{theorem}

It is shown in \cite{benanav87} that associative-commutative matching
can be done in polynomial time if $t_1$ is a DO-term.
We can extend their algorithm as below.
For extension, it is enough to add a condition in their algorithm that
$f((t_1)_{u_1},\ldots,(t_1)_{u_p})$ and
$f((t_2)_{v_1},\ldots,(t_2)_{v_q})$ can be unified
if $(t_1)_{u_i}$ and $(t_2)_{v_j}$ are variables for some $i,j$.

\begin{proposition}
Associative-commutative unification
can be done in polynomial time if both $t_1$ and $t_2$ are DO-terms.
\label{prop:do-ac-unif}
\end{proposition}

\newpage

\section*{Appendix}

\subsection*{A1. Latter Part of the Proof of Theorem~\ref{thm:edhard}}

Next we consider the case in which the number of occurrences of each
variable is bounded by 3.
For that purpose, we rename $j$-th occurrence of $x_i$ in $s^1$
by $x_{i,j}$ and then we append $t^1 t^2 \cdots t^l$ and
$u^1 u^2 \cdots u^l$ to the end of $s^1$ and $s^2$, respectively,
where $t^i$ and $u^i$ are defined as follows:
\begin{eqnarray*}
t^i & = & \# x_{i,1} y_{i,1} y_{i,1} x_{i,2} y_{i,2} y_{i,2} x_{i,3} y_{i,3} y_{i,3} x_{i,4} \cdots x_{i,l-1} y_{i,l-1} y_{i,l-1} x_{i,l}, \\
u^i & = & \# z_{i,1} z_{i,1} z_{i,2} z_{i,2} z_{i,2} z_{i,3} z_{i,3} z_{i,3} z_{i,4} z_{i,4} \cdots z_{i,l-1} z_{i,l-1} z_{i,l} z_{i,l}.
\end{eqnarray*}
Let the resulting strings be $\hat{s}_1$ and $\hat{s}_2$.
Then, it is seen that
there exists an LCS of length $|\hat{s}^1|$ iff
$\hat{d}_S(\hat{s}^1,\hat{s}^2) = (\sum_i |s_i|)-lk$.
Since the number of variables is $3lk$, which is still a polynomial of
$l$ and $k$, it is an FPT reduction.
\qed

\subsection*{A2. Proof of Theorem~\ref{thm:tree-ed-do}}

Let $F_1$ and $F_2$ be ordered forests.
Let $T_1$ (resp., $T_2$) be the rightmost tree
of $F_1$ (resp., $F_2$).
It is well-known \cite{bille05} that
the rooted ordered tree edit distance can be
computed by the following dynamic programming procedure:
\begin{eqnarray*}
D[\epsilon,\epsilon] & \leftarrow & 0,\\
D[F_1,\epsilon] & \leftarrow & D[F_1-r(T_1),\epsilon] + \delta(r(T_1),-),\\
D[\epsilon,F_2] & \leftarrow & D[\epsilon,F_2-r(T_2)] + \delta(-,r(T_2)),\\
D[F_1,F_2] & \leftarrow & \min \left\{
\begin{array}{l}
D[F_1-r(T_1),F_2]+\delta(r(T_1),-),\\
D[F_1,F_2-r(T_2)]+\delta(-,r(T_2)),\\
D[F_1-T_1,F_2-T_2]+D[T_1-r(T_1),T_2-r(T_2)]+\delta(r(T_1),r(T_2)).
\end{array}
\right.
\end{eqnarray*}
where $\delta(x,y)$ is a delta function
(i.e., $\delta(x,x)=0$, otherwise $\delta(x,y)=1$),
$\epsilon$ denotes the empty tree,
$F-v$ (resp., $F-T$) denotes the forest obtained by deleting $v$ (resp., $T$)
from $F$, and $D[t_1,t_2]$ gives the required edit distance.
In order to cope with DO-variables,
it is enough to add the following in taking the minimum at the recursion
of computing $D[F_1,F_2]$:
\begin{eqnarray*}
D[F_1-T_1,F_2-T_2],~~~~~ \mbox{if $T_1$ or $T_2$ consists of a variable node}.
\end{eqnarray*}
Then, it is clear that
the order of the time complexity is the same (i.e., $O(m^2 n^2)$)
as that of the original one.
\qed

\subsection*{A3. String Matching with Variable Length Don't Cares}

Let $a=a_1 \ldots a_p$ and $b=b_1 \ldots b_q$ be strings including
variables length don't care characters `*'.
The following is a pseudo-code for deciding whether $a$ matches $b$,
where $E[i,j]=1$ iff $a=a_1 \ldots a_i$ matches $b=b_1 \ldots b_j$.

\begin{rm}
\begin{tabbing}
\quad \= \quad \= \quad \= \quad \= \quad \= \quad \= \quad \= \kill
\> Procedure $StrMatchVDC(a,b)$\\
\> \> {\bf for all} $i,j \in \{0,\ldots,p\} \times \{0,\ldots,q\}$ {\bf do} 
$E[i,j] \leftarrow 0$;\\
\> \> $E[0,0] \leftarrow 1$;\\
\> \> {\bf for} $i=1$ {\bf to} $p$ {\bf do}\\
\> \> \> {\bf for} $j=1$ {\bf to} $q$ {\bf do}\\
\> \> \> \> {\bf if} $a_i = *$ and $b_j = *$ {\bf then}\\
\> \> \> \> \> $E[i,j] \leftarrow \max \{ \max_{i'<i} \{ E[i',j-1] \},
\max_{j'<j} \{ E[i-1,j'] \} \}$ \\
\> \> \> \> {\bf else if} $a_i=*$ {\bf then}
$E[i,j] \leftarrow \max_{j'<j} \{ E[i-1,j'] \}$\\
\> \> \> \> {\bf else if} $b_j=*$ {\bf then}
$E[i,j] \leftarrow \max_{i'<i} \{ E[i',j-1] \}$\\
\> \> \> \> {\bf else} if $a_i$ matches $b_j$ {\bf then}
$E[i,j] \leftarrow E[i-1,j-1]$\\
\> \> \> \> {\bf else} 
$E[i,j] \leftarrow 0$;\\
\> \> {\bf if} $E[p,q]=1$ {\bf then} {\bf return} true {\bf else} {\bf return} false.
\end{tabbing}
\end{rm}

It is obvious that this algorithm works in $O(p^2 q^2)$ time.

\subsection*{A4. Proof of Theorem~\ref{thm:com-unif-poly}}

The correctness of $TestCommutIdent(r_1,r_2,G(V,E))$ follows from the fact
that $f_1(t_1,t_2)$ matches $f_2(t_1',t_2')$ iff
$f_1$ and $f_2$ are the identical function symbols and
either $(t_1,t_2)$ matches $(t_1',t_2')$ or 
$(t_1,t_2)$ matches $(t_2',t_1')$.
It is clear that this procedure works in $O(mn)$ time.
Therefore, commutative match of two variable-free terms can be tested in 
polynomial time.

Next we consider the correctness of $CommutUnify(t_1,t_2)$
(see also Fig.~\ref{fig:dag}).
It is straight-forward to see that if there exists $M$ which returns `true',
$t_1$ and $t_2$ are commutatively unifiable and such a mapping
gives a substitution $\theta$ satisfying $t_1 \theta = t_2 \theta$.
Conversely, suppose that
$t_1$ and $t_2$ are commutatively unifiable.
Then,
there exist unifiable non-commutative terms $t_1'$ and $t_2'$
that are obtained by exchanging
left and right arguments in some terms in $t_1$ and $t_2$.
Let $\theta$ be the substitution satisfying $t_1' \theta = t_2' \theta$.
Then, $t_1 \theta = t_2 \theta$ holds.
We assign distinct constants to variables appearing in $t_1 \theta$.
We also construct a mapping from the remaining variables to $N(t_1) \cup N(t_2)$
by regarding $x/t \in \theta$ as a mapping of $x$ to $t$.
We construct $G(V,E)$ according to this mapping.
Then, it is obvious that $CommutIdent(r_1,r_2,G(V,E))=$true holds.

Since the number of possible mappings is bounded by $|m+n|^k$
where $k$ is the number of variables in $t_1$ and $t_2$,
$CommutUnify(t_1,t_2)$ works in polynomial time.
\qed

\begin{figure}[ht]
\begin{center}
\includegraphics[width=8cm]{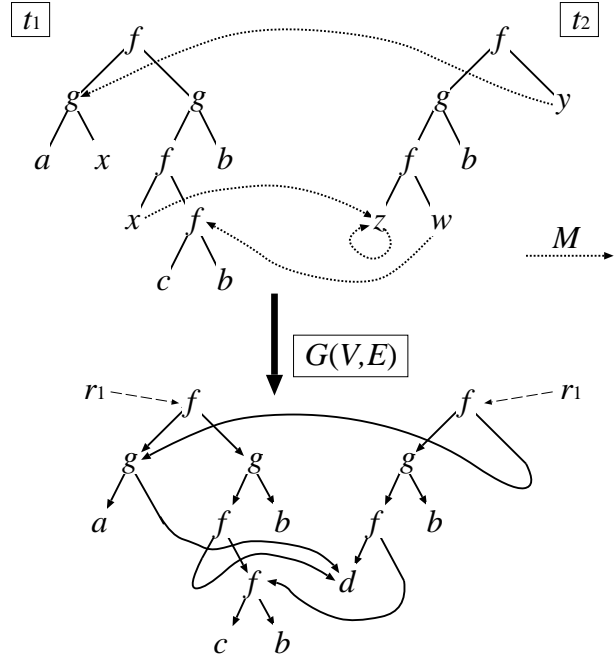}
\caption{Example of DAG $G(V,E)$ for the algorithm and proof of Theorem~\ref{thm:com-unif-poly}.}
\label{fig:dag}
\end{center}
\end{figure}

\subsection*{A5. Proof of Theorem~\ref{thm:ac-hard}}

We show an FPT-reduction from LCS.
Let $(S,l)$ be an instance of LCS
where $S = \{s_1,\ldots,s_k\}$ is a set of strings and $l$ is an integer.

Different from the previous proofs,
we cannot represent the order of letters directly.
Therefore, we represent the position of each letter by
the size of the corresponding term.

Let $f_1,f_2,f_3,f_4$ be distinct functional symbols
and $a$ be a constant not appearing in $S$.
For each $s_i[j]$,
we define the term $\hat{s}_i[j]$ by
\[
\hat{s}_i[j] = f_1(s_i[j],f_2(\overbrace{a,a,\cdots,a}^{j})).
\]
Then, we define the term $t_2$ by
\begin{eqnarray*}
t_2 & = & f_3(f_4(\hat{s}_1[1],\hat{s}_1[2],\cdots,\hat{s}_1[|s_1|]),
f_4(\hat{s}_2[1],\hat{s}_2[2],\cdots,\hat{s}_2[|s_2|]),\\
& & \cdots,
f_4(\hat{s}_k[1],\hat{s}_k[2],\cdots,\hat{s}_k[|s_k|])).
\end{eqnarray*}

Next, we define the term $t^i_j$ ($i=1,\ldots,k$, $j=1,\ldots,l)$ by
\[
t^i_j = f_1(x_j,f_2(y_{i,1},y_{i,2},\cdots,y_{i,j})),
\]
where $x_j$ and $y_{i,h}$s are variables.
Then, we define $t^i$ ($i=1,\ldots,l$) and $t_1$ by
\begin{eqnarray*}
t^i & = & f_4(z_i,t^i_1,t^i_2,\cdots,t^i_k),\\
t_1 & = & f_3(t^1,t^2,\cdots,t^l),
\end{eqnarray*}
where $z_i$ is a variable.

We show that $t_1$ and $t_2$ are unifiable if and only if
LCS of $S$ has length at least $l$.
First we show the `if' part.
Let $s_c$ be a common subsequence of $S$ such that $|s_c|=l$.
We consider a partial substitution $\theta'$ defined by
\[
\theta' = \{x_1/s_c[1],x_2/s_c[2],\ldots,x_l/s_c[l]\}.
\]
Then, it is straight-forward to see
that $\theta'$ can be extended to a substitution $\theta$ such
that $t_2 = t_1 \theta$.

Conversely, suppose that there exists a substitution $\theta$ satisfying
$t_2 = t_1 \theta$.
We can see from the construction of $t_1$ and $t_2$ that
if $x_j$ matches $s_i[h]$ and $x_{j'}$ matches $s_i[h']$
for some $i \in \{1,\ldots,k\}$ where $j < j' \leq l$,
then $h < h'$ must hold.
Let $x_j/a_j \in \theta$ ($j=1,\ldots,l$).
Then, we can see from the above property that
$s_c = a_1 a_2 \ldots a_l$ is a common subsequence of $S$.

Since the reduction can be done in polynomial time and
the number of variables is bounded by a polynomial of $k$ and $l$,
the theorem holds.
\qed

\end{document}